%% file: main.tex
\newif\ifanonym
\anonymtrue
\anonymfalse %

\documentclass[12pt]{article}
\usepackage[utf8]{inputenc}

\usepackage[letterpaper, margin=1in]{geometry}

\input{defs.tex} %

\title{Exact Flow Sparsification Requires Unbounded Size%
  \footnote{The first version of this paper proved a weaker statement of \cref{:thm:refutation-of-seymour's-conjectures} with $4$ commodities.
    The current statement has only $3$ commodities,
    and now fully refutes Seymour's conjectures.
  }
}

\ifanonym
\author{Anonymous Authors}
\else
\author{
Robert Krauthgamer%
\thanks{%
    Work partially supported by ONR Award N00014-18-1-2364,
    the Israel Science Foundation grant \#1086/18,
    the Weizmann Data Science Research Center,
    and a Minerva Foundation grant.
    Email: \texttt{robert.krauthgamer@weizmann.ac.il}
}
\qquad Ron Mosenzon%
\thanks{%
  Email: \texttt{ron.mosenzon@weizmann.ac.il}
}
\\
Weizmann Institute of Science
}
\fi

\begin{document}

\setcounter{page}{1}
\maketitle

\begin{abstract}
Given a large edge-capacitated network $G$ and a subset of $k$ vertices called \emph{terminals},
an (\emph{exact}) \emph{flow sparsifier} is a small network $G'$
that preserves (exactly) all multicommodity flows that can be routed between the terminals.
Flow sparsifiers were introduced by Leighton and Moitra [STOC 2010],
and have been studied and used in many algorithmic contexts.

A fundamental question that remained open for over a decade,
asks whether every $k$-terminal network admits an exact flow sparsifier
whose size is bounded by some function $f(k)$
(regardless of the size of $G$ or its capacities).
We resolve this question in the negative
by proving that there exist $6$-terminal networks $G$ 
whose flow sparsifiers $G'$ must have arbitrarily large size. 
This unboundedness is perhaps surprising,
since the analogous sparsification that preserves all terminal cuts
(called \emph{exact cut sparsifier} or \emph{mimicking network}) 
admits sparsifiers of size $f_0(k)\leq 2^{2^k}$ 
[Hagerup, Katajainen, Nishimura, and Ragde, JCSS 1998].
    
We prove our results by analyzing the set of all feasible demands in the network,
known as the \emph{demand polytope}.
We identify an invariant of this polytope,
essentially the slope of certain facets,
that can be made arbitrarily large even for $k=6$,
and implies an explicit lower bound on the size of the network. 
We further use this technique to answer, again in the negative, 
an open question of Seymour [JCTB 2015] regarding flow-sparsification
that uses only contractions and preserves the infeasibility of one demand vector. 
\end{abstract}

\newpage

\section{Introduction}

Graph compression is a powerful paradigm in the design of graph algorithms. Where one reduces the size of the graph before performing heavy computation on it, while preserving properties of the graph that are needed for downstream computation. This paradigm is known to be extremely useful for faster computation, smaller memory requirement, as well as better accuracy when using approximation algorithms.

We study a form of graph compression known as \emph{vertex sparsification}, where one is given an undirected network $G = (V,E,c)$ with edge capacities $c:E \to \bbR_{\geq 0}$, and a small set $T \subseteq V$ of vertices from $G$ called \emph{terminals}.
The goal is to create a small network $G' = (V',E',c')$ called a vertex sparsifier, which may have different vertices, edges, and capacities than $G$, but includes the same set of terminals $T \subseteq V'$, and preserves some desired relationship between the terminals.
We study flow sparsifiers,
where the relationship of interest is the set of feasible multicommodity flows,
as introduced by Leighton and Moitra~\cite{LM10}. 

Let us recall the relevant definitions.
Given $G$ and $T$ as above,
a pair of terminals $i \in \binom{T}{2}$ is referred to as a \emph{commodity},
and the two terminals of each commodity are arbitrarily assigned to be the source and sink of that commodity.
A \emph{demand vector} $\mathbf{d} \in \bbR_{\geq 0}^{\binom{T}{2}}$
assigns a non-negative value $d_i$ to each commodity $i \in \binom{T}{2}$.
A multicommodity flow $f$ that realizes $\mathbf{d}$ is a collection of $\binom{|T|}{2}$ flow functions, namely, $f_i:E \to \bbR$ for each commodity $i \in \binom{T}{2}$,
such that $f_i$ ships $d_i$ units of flow from the source to the sink of commodity $i$ (satisfying flow conservation at every other vertex), and together these flows satisfy the the following capacity constraint
\[
 \forall e \in E, \qquad \sum_{i \in \binom{T}{2}} |f_i(e)| \leq c(e) .
\]
Throughout, one fixes an arbitrary orientation of the edges of $G$, 
so that the sign of the flow $f_i(e)$ determines its direction. 
A demand vector $\mathbf{d}$ is called \emph{feasible} in $G$
if there exists a multicommodity flow in $G$ that realizes it.
The set of all feasible demands in $G$, 
called the \emph{demand polytope} (or \emph{throughput polytope}, see e.g.~\cite{RBC07}) 
is defined as 
\[
 \cD(G) \defeq \{\mathbf{d} \in \bbR_{\geq 0}^{\binom{T}{2}} : \text{$\mathbf{d}$ is feasible in $G$}\} .
\]
It is not hard to see that $\cD(G)$ is indeed a polytope. 

\begin{definition} [Exact Flow Sparsifier] \label{:def:flow-sparsifier}
A network $G'=(V',E',c')$ is called an \emph{exact flow sparsifier}
(or a flow sparsifier of quality $1$) of a network $G=(V,E,c)$
if they share the same terminal set $T$ and
\[
 \cD(G') = \cD(G) .
\]
\end{definition}

\subsection{Results} \label{:subsec:results}

We answer the following problem, that is open since \cite{LM10} and is well-known to experts:
Does every $k$-terminal network $G$ have an exact flow sparsifier of size bounded by some function $f(k)$?
Interestingly, this question is open even if the sparsifier is allowed $O(1)$-approximation,
stated explicitly in \cite[Section 7]{Moitra11} and \cite[Section 1]{Chuzhoy12}.
Our main result resolves this question in the negative (for exact sparsifiers) by showing that even $6$-terminal networks can require arbitrarily large sparsifiers.

\begin{restatable}{theorem}{MainResult} \label{:thm:main-result}
For every $k \geq 6$, there exists an infinite family $\{G_l\}_{l \in \bbN}$ of $k$-terminal networks,
such that every exact flow sparsifier of every $G_l$ must have at least $l$ vertices.
\end{restatable}

This theorem %
reveals a stark contrast with the analogous sparsification that preserves all terminal cuts,
called \emph{exact cut sparsifiers} (or \emph{mimicking networks}),
which admit constructions (exact cut sparsifiers)
of size $f_0(k) \leq 2^{2^k}$ \cite{HKNR98, KR14}.
Furthermore, it is known that for networks $G$ with at most $4$ terminals, an exact cut sparsifier of $G$ is also an exact flow sparsifier of $G$ (because there is no flow-cut gap, see e.g.~\cite{AGK14}),
and therefore such networks admit exact flow sparsifiers of size bounded by $f_0(4)$.
This means that the condition $k \geq 6$ in the theorem is almost tight.
We further note that the family of networks used in \cref{:thm:main-result} excludes a fixed minor, namely, $K_7$. 
We discuss these issues further in \cref{:sec:further-directions}.

We further use our techniques to answer a question of Seymour \cite{Seymour15}
regarding a related notion of flow-sparsification,
that uses only contractions and preserves the infeasibility of one demand vector. 
This is discussed in \cref{:subsec:seymour's-conjectures}.

\subsection{Techniques and Proof Outline}

Our proof of \cref{:thm:main-result} is based on identifying a numerical feature of the demand polytope that forces a large sparsifier network. Some natural candidates are the maximum size of feasible demands, or the bit complexity required to represent the polytope. These candidates fail because even a small sparsifier can have arbitrarily large capacities, and thus can also represent an unbounded number of bits. To circumvent this problem, we identify a feature that can be upper bounded in a manner independent of the capacities, yet can grow arbitrarily large as the network size increases (even as the number of terminals remains fixed).

In essence, $\cD(G)$ represents all possible tradeoffs between the values of the commodities. Since this is too complicated, we focus on the tradeoff between just two commodities, at the neighborhood of a single point $\mathbf{d} \in \bbR_{\geq 0}^{\binom{T}{2}}$. We then use this tradeoff as our numerical feature. 
We remark that our proof can be adapted to rely on the number of facets of the polytope,
a numerical feature that is a widely used in polyhedral theory.%
\footnote{However, it seems more difficult to prove that the facet count can be arbitrarily large when the number of terminals is constant,
as our proof method is essentially a reduction to tradeoff argument. 
Furthermore, we believe that our tradeoff feature is more versatile for analyzing demand polytopes, as exemplified by our refutation of Seymour's conjecture \cite{Seymour15} 
in \Cref{:subsec:seymour's-conjectures}, 
which does not seem to have an analogous proof using facet counting.
}

\begin{definition}[Tradeoff]
Let $\mathbf{d} \in \bbR_{\geq 0}^{\binom{T}{2}}$ be a point in the demand polytope $\cD(G)$ of a network $G$ with terminal set $T$.
We say that $\tau \geq 0$ is a \emph{feasible tradeoff} between commodity $i \in \binom{T}{2}$ and commodity $j \in \binom{T}{2}$ at the point $\mathbf{d}$ if there exists a vector $\mathbf{d}^{\Delta} \in \bbR^{\binom{T}{2}}$ with entries $d^{\Delta}_j < 0$ and $d^{\Delta}_i = - \tau d^{\Delta}_j$ while all other entries are zero, such that $\mathbf{d} + \mathbf{d}^{\Delta} \in \cD(G)$.%
\footnote{Observe that if $\tau > 0$ is a feasible tradeoff then also every $\tau' \in [0,\tau]$ is a feasible tradeoff, because the polytope $\cD(G)$ is down-monotone.}
The \emph{maximum tradeoff} between commodities $i$ and $j$ at the point $\mathbf{d}$ is defined as
\[
  \tau_{\mathbf{d},i,j}
  \defeq
  \sup\{\text{feasible tradeoff $\tau \geq 0$ between commodities $i$ and $j$ at point $\mathbf{d}$}\} .
\]
\end{definition}

Observe that $\tau_{\mathbf{d},i,j}$ can be infinite.
However, we are interested in finite $\tau_{\mathbf{d},i,j}$, 
and indeed the aforementioned numerical feature of the demand polytope
is the largest finite $\tau_{\mathbf{d},i,j}$ over all $\mathbf{d},i,j$. 
The maximum tradeoff $\tau_{\mathbf{d},i,j}$ has a geometric interpretation in terms of the demand polytope $\cD(G)$. Consider the two-dimensional plane of points that differ from $\mathbf{d}$ only in coordinates $i$ and $j$, then the intersection of this plane with $\cD(G)$ is a two-dimensional polytope.
When $\mathbf{d}$ is strictly inside this polytope, $\tau_{\mathbf{d},i,j}$ is infinite. 
However, when $\mathbf{d}$ lies on a one-dimensional face of this two-dimensional polytope, $\tau_{\mathbf{d},i,j}$ is exactly the slope of this face, 
see \cref{:fig:tradeoff-definition} for illustration.
Furthermore, if the point $\mathbf{d}$ is on exactly one facet of $\cD(G)$,
say a facet $\sum_{i' \in \binom{T}{2}} a_{i'} d_{i'} \leq b$, 
then the maximum tradeoff is $\frac{a_j}{a_i}$,
which is essentially a \quotes{directional slope} of this facet. 

\begin{figure}[ht]
\includegraphics[scale=0.8]{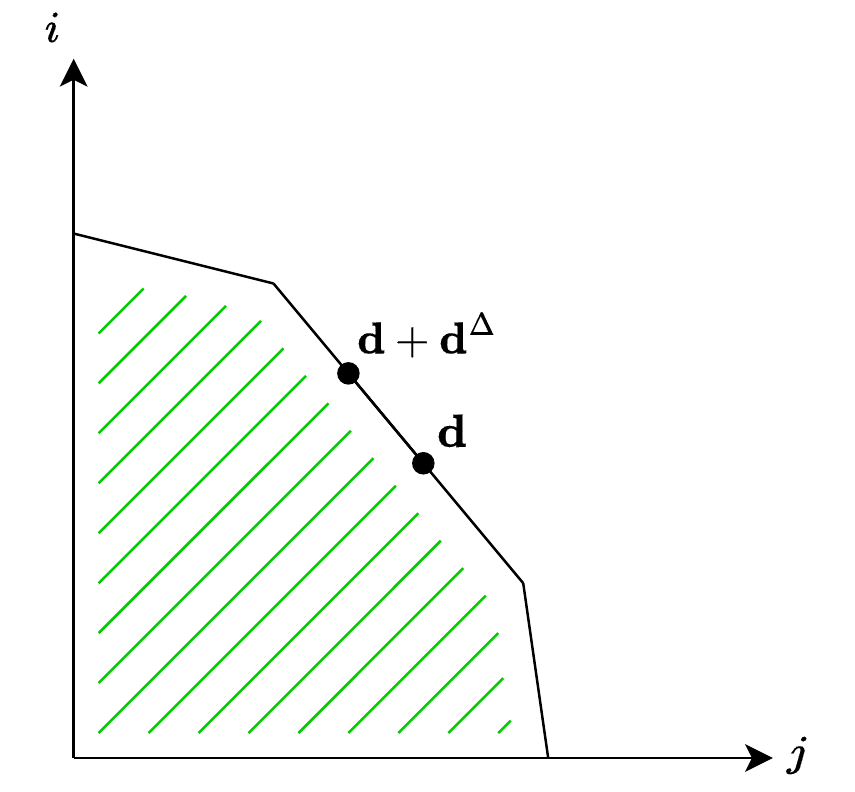}
\centering
\caption{The two-dimensional plane of points that differ from $\mathbf{d}$ only in coordinates $i$ and $j$. 
The region inside $\cD(G)$ is shown in green stripes. The vector $\mathbf{d}^{\Delta}$ attains the maximum tradeoff $\tau_{\mathbf{d},i,j}$, which is the slope of the face on which $\mathbf{d}$ lies.
}
\label{:fig:tradeoff-definition}
\end{figure}

\cref{:thm:main-result} follows from the next two propositions,
which we prove 
in \cref{:sec:proof-of-tradeoff-upperbound} and \cref{:sec:unbounded-tradeoff-family}, respectively.

\begin{restatable}[Tradeoff Bound]{proposition}{TradeoffUpperBoundProposition} \label{:prop:tradeoff-upper-bound}
Let $G$ be an $n$-vertex network with terminal set $T$,
let $\mathbf{d} \in \cD(G)$ be a point, 
and let $i^* \in \binom{T}{2}$ and $j^* \in \binom{T}{2}$ be two distinct commodities.
If the maximum tradeoff $\tau_{\mathbf{d},i^*,j^*}$ is finite,
then it is bounded by $\tau_{\mathbf{d},i^*,j^*} \leq (n^4)!$.
\end{restatable}

\begin{restatable}[Graph Family with Unbounded Tradeoff]{proposition}{UnboundedTradeoffFamilyProposition} \label{:prop:unbounded-tradeoff-family}
For every $M > 0$, there exists a $6$-terminal network $G_M$, a point $\mathbf{d} \in \cD(G_M)$, and two distinct commodities $i^* \in \binom{T}{2}$ and $j^* \in \binom{T}{2}$, for which the maximum tradeoff $\tau_{\mathbf{d},i^*,j^*}$ is finite and $\tau_{\mathbf{d},i^*,j^*} \geq M$.
\end{restatable}

\begin{proof}[Proof of \cref{:thm:main-result}]
For every $l \in \bbN$, choose $M > ((l-1)^4)!$
and consider the network $G_M$ from \cref{:prop:unbounded-tradeoff-family}. 
It is then guaranteed by \cref{:prop:tradeoff-upper-bound}
that this network has no exact flow sparsifier with less than $l$ vertices.
\end{proof}

\subsection{Seymour's Conjectures}
\label{:subsec:seymour's-conjectures}

Seymour~\cite{Seymour15} investigated the following flow-sparsification problem.

\begin{restatable}{problem}{SeymoursProblem} \label{:prob:seymour's-problem}
Given a network $G = (V,E,c)$ with integral capacities
together with an infeasible integral demand vector $\mathbf{d}^* \notin \cD(G)$
on the set of commodities $\binom{V}{2}$,
the goal is to apply on $G$ edge contractions
and obtain a small network $G'$ in which the demand $\mathbf{d}^*$ is still infeasible.%
\footnote{Here, the contractions affect demands in the natural way: 
when an edge $(u,v)$ is contracted, 
demands that involve its endpoints $u$ or $v$ carry over to contracted vertices.
Notice that all vertices are viewed as terminals, but this is inconsequential.
}
Such a network $G'$ is said to \emph{compatible} with the problem instance $(G,\mathbf{d}^*)$.
\end{restatable}

Seymour \cite{Seymour15} proved the existence of compatible networks $G'$ for $(G,\mathbf{d}^*)$ with size bounded with respect to the \emph{total demand},
i.e., the sum of the entries in the demand vector $\mathbf{d}^*$,
and asked whether the size of a compatible network can be bounded independently of this parameter.
Seymour actually posed two questions (see Section 1 and Table 1 in~\cite{Seymour15}), which we rephrase and present as two conjectures:
\begin{itemize} 
\item
  There exists $N>0$,
  such that every instance of \cref{:prob:seymour's-problem}
  has a compatible network $G'$ with at most $N$ vertices.
\item
  There exists $N_3>0$, such that every instance of \cref{:prob:seymour's-problem} with $\norm{\mathbf{d}^*}_0 \leq 3$
  has a compatible network $G'$ with at most $N_3$ vertices.
  (As usual, $\norm{\cdot}_0$ is the number of non-zero entries in a vector,
  and in our case it is the number of non-trivial commodities.)
\end{itemize}
The first conjecture clearly implies the second one, and is thus stronger.

We refute both conjectures.
We present below the formal statement and its proof sketch, 
and provide the full proof in \Cref{:sec:proofSeymour}. 
\begin{restatable}{theorem}{RefutationOfSeymoursConjectures} \label{:thm:refutation-of-seymour's-conjectures}
For every $N > 0$, there exists an instance of \cref{:prob:seymour's-problem} with $\norm{\mathbf{d}^*}_0 \leq 3$, for which every compatible network must have more than $N$ vertices.
\end{restatable}

\begin{proof}[Proof sketch]
In this sketch, we ignore the issues of bounded $\norm{\mathbf{d}^*}_0$ and of integrality.
Our goal in the proof is to construct an instance for which every compatible network must have a large tradeoff (i.e. have a point with a large finite maximum tradeoff), and thus, by \cref{:prop:tradeoff-upper-bound}, every compatible network must have a large number of vertices.
To do this, we rely on two observations. The first one is that, because contraction operations can only eliminate/relax constraints on the flows, every network $G'$ that is compatible with an instance $(G,\mathbf{d}^*)$ must have $\cD(G') \supseteq \cD(G)$.
The second observation is that the maximum tradeoff has an alternative characterization as follows:
\begin{restatable}{claim}{AlternativeCharacterizationOfTradeoff} \label{:cl:alternative-characterization-of-tradeoff}
Let $G$ be a network, and let $i^*$ and $j^*$ be two distinct commodities in $\binom{V}{2}$. Then, for every $\tau > 0$, the following two statements are equivalent.
\begin{enumerate}
\item
  There exists a point $\mathbf{d} \in \cD(G)$ with a finite $\tau_{\mathbf{d},i^*,j^*} > \tau$.
\item
  There exists a point $\mathbf{d}^* \notin \cD(G)$, a point $\mathbf{\hat{d}} \in \cD(G)$ that results from $\mathbf{d}^*$ by zeroing coordinate $i^*$, and a vector $\mathbf{d}^{\Delta*} \in \bbR^{\binom{T}{2}}$ with $d^{\Delta*}_{j^*} < 0$ and $d^{\Delta*}_{i^*} = - \tau d^{\Delta*}_{j^*}$ while all other entries are zero, such that $\mathbf{d}^* + \mathbf{d}^{\Delta*} \in \cD(G)$. 
  (See \cref{:fig:alternative-tradeoff-characterization} for illustration.) 
\end{enumerate}
\end{restatable}

\cref{:cl:alternative-characterization-of-tradeoff} tells us that in order to preserve a large tradeoff, it is enough to keep one specific point outside of the demand polytope and keep two specific points inside the demand polytope. Together with our first observation, this implies that by creating an instance from a network with a large tradeoff and using the right $\mathbf{d}^*$, we can force all compatible networks to have large tradeoff. The theorem follows from \cref{:prop:unbounded-tradeoff-family} and \cref{:prop:tradeoff-upper-bound}.
\end{proof}

\begin{figure}[htb]
\includegraphics[scale=0.8]{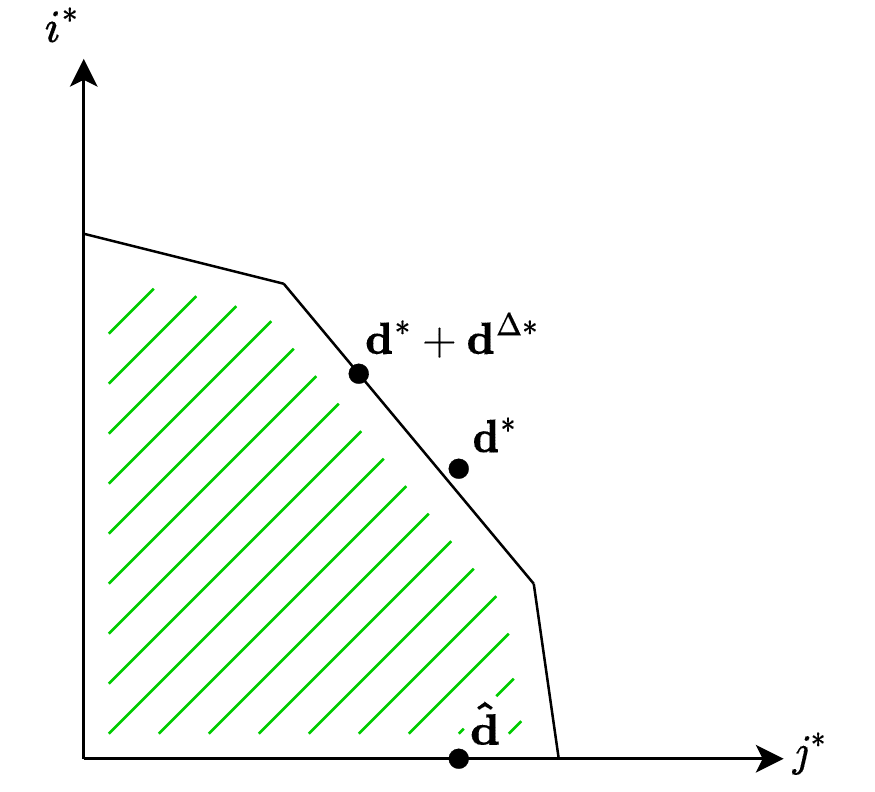}
\centering
\caption{The condition for a large finite maximum tradeoff from 
statement 2 in \cref{:cl:alternative-characterization-of-tradeoff}. 
The two-dimensional plane corresponds to points that differ from $\mathbf{d}^*$ only in coordinates $i^*$ and $j^*$,
and the region inside $\cD(G)$ is shown in green stripes.
}
\label{:fig:alternative-tradeoff-characterization}
\end{figure}

\subsection{Related Work}

There is a large body of work on flow sparsification.
It has began with Leighton and Moitra~\cite{LM10},
who studied approximate sparsifiers
(in contrast to our focus on exact sparsifiers)
that are supported \emph{only} on the terminals.
This setting was studied further by Charikar et al.~\cite{CLLM10}, Makarychev and Makarychev~\cite{MM16}, and Englert et al.~\cite{EGKRTT14},
who tightened the approximation factor, called quality.
Chuzhoy~\cite{Chuzhoy12, Chuzhoy16} introduced a more general setting of approximate flow sparsification, 
where the sparsifier can have non-terminal vertices, 
and she constructed sparsifiers whose size bound depends on the total capacity of all edges incident to the terminals, as well as a version of these sparsifiers that preserves both fractional and integral flows.
Several improved bounds are known for important graph families,
like quasi-bipartite graphs~\cite{AGK14}, trees~\cite{GR16}, and planar~\cite{EGKRTT14}. 

Another concept that is closely related to flow sparsification is \emph{cut sparsification}.
Exact cut sparsifiers were introduced by Hagerup, Katajainen, Nishimura, and Ragde~\cite{HKNR98}, 
under the name \emph{mimicking networks}, 
and tightened bounds for the general case and special graph families
were proved in \cite{CSWZ00,KR13,KR14,KPZ19,GHP20,KR20}. 
Approximate cut sparsifiers were introduced by Moitra~\cite{Moitra09},
and were often studied in conjunction with approximate flow sparsifier,
because of the flow-cut gap.
Another closely related notion is that of
(exact) vertex-cut sparsifiers \cite{KW20, HLW21, CDLKLPSV21, Liu20, BK22}.
There are many other notions of sparsifiers,
e.g., preserving other graph properties like spectrum or distances,
or sparsifiers that have few edges (but the same vertex set),
and listing all of them here would be prohibitive.

\subsection{Preliminaries}

Throughout this paper, we use tilde to denote formal variables, such as $\mathbf{\tilde{x}}$.

Given a matrix $A \in \bbR^{m_1 \times m_2}$ with rows $\mathbf{a_1}^\top,\ldots,\mathbf{a_{m_1}}^\top$, as well as a column vector $\mathbf{b} = (b_1,...,b_{m_1}) \in \bbR^{m_1}$, 
the \emph{feasibility linear program} $A \mathbf{\tilde{x}} \geq \mathbf{b}$ is the problem of assigning values $\mathbf{x} = (x_1,\ldots,x_{m_2})$ to the formal variables $\mathbf{\tilde{x}} = (\tilde{x}_1,\ldots,\tilde{x}_{m_2})$, 
such that each of the inequalities $\{\mathbf{a_i}^\top \mathbf{x} \geq b_i\}_{i \in [m_1]}$ holds. 
Such an assignment $\mathbf{x}$ is called a \emph{feasible solution} of the linear program,
and the set of all feasible solutions is called the \emph{feasible region} of the linear program.

\begin{definition}
Let $A \mathbf{\tilde{x}} \geq \mathbf{b}$ be a linear program with $m_2$ variables, 
and let $S \subseteq [m_2]$ be a subset of the variables. 
Then, the \emph{projection of the linear program} onto the the variables of $S$ is
\[
 \{\mathbf{x}_{|S} : \mathbf{x} \in \bbR^{m_2}, A \mathbf{x} \geq \mathbf{b}\} ,
\]
where $\mathbf{x}_{|S}$ denotes the restriction of $\mathbf{x}$ to the coordinates in $S$.
\end{definition}

The following theorem is folklore. 

\begin{theorem} \label{:thm:folklore-theorem}
Let $A \mathbf{\tilde{x}} \geq \mathbf{b}$ be a linear program with $m_1$ constraints and $m_2$ variables, 
and suppose that $\max\{\mathbf{c}^\top \mathbf{x} : A \mathbf{x} \geq \mathbf{b}\}$ 
is finite. 
Then, there exists a subset $S_1 \subseteq [m_1]$ of constraints and a subset $S_2 \subseteq [m_2]$ of variables,
such that the system of linear equalities $A_{|S_1 \times S_2} \mathbf{\tilde{x}}_{|S_2} = \mathbf{b}_{|S_1}$ has a unique solution, 
and such that extending this solution to the variables outside $S_2$ by setting them to $0$ results in an optimal solution for the objective $\mathbf{c}^\top \mathbf{x}$ under constraint $A \mathbf{x} \geq \mathbf{b}$.%
\end{theorem}

\section{Tradeoff Bound (Proof of \cref{:prop:tradeoff-upper-bound})}
\label{:sec:proof-of-tradeoff-upperbound}

This section proves \cref{:prop:tradeoff-upper-bound}.
For convenience, we first recall its statement. 
\TradeoffUpperBoundProposition*

Our proof of \cref{:prop:tradeoff-upper-bound} relies on the following well-known lemma. For completeness, a proof of this lemma is presented in the appendix.
\begin{lemma} \label{:lem:existance-of-LP}
$\cD(G)$ can be described as the projection of a linear program $A \mathbf{\tilde{x}} \geq \mathbf{b}$ with at most $n^4$ variables onto $\binom{|T|}{2}$ of its variables that represent the commodities, such that all entries of $A$ are from $\{-1,0,1\}$.
\end{lemma}
Let $G, T, \mathbf{d}$, $i^*$ and $j^*$ be as in \cref{:prop:tradeoff-upper-bound},
and let $A \mathbf{\tilde{x}} \geq \mathbf{b}$ be the linear program promised by \cref{:lem:existance-of-LP}. According to the lemma, $\cD(G)$ is the projection of the program onto $\binom{|T|}{2}$ of its variables. Let $L: \bbR^{m_2} \to \bbR^{\binom{T}{2}}$ be the projection onto these variables. Let $m_1$ and $m_2$ be the number of constraints and the number of variables in the aforementioned linear program. Then, $m_2 \leq n^4$.
Since $\mathbf{d} \in \cD(G)$, there exists a feasible solution $\mathbf{x}$ of $A \mathbf{\tilde{x}} \geq \mathbf{b}$ such that $L(\mathbf{x}) = \mathbf{d}$.
Let $EQ(\mathbf{x}) \subseteq [m_1]$ denote the set of inequalities in $A \mathbf{\tilde{x}} \geq \mathbf{b}$ that are tight at this solution $\mathbf{x}$, and let $A_{EQ(\mathbf{x})} \mathbf{\tilde{x}} \geq \mathbf{b}_{EQ(\mathbf{x})}$ be the linear program defined by this subset of the constraints.
We use the following claim to complete the proof of \cref{:prop:tradeoff-upper-bound}.

\begin{claim} \label{:cl:intermediate-claim}
A value $\tau \geq 0$ is a feasible tradeoff at the point $\mathbf{d}$ between commodities $i^*$ and $j^*$ iff there exists a solution $\mathbf{x}^{\Delta} \in \bbR^{m_2}$ of the linear program $A_{EQ(\mathbf{x})} \mathbf{\tilde{x}} \geq \mathbf{0}$, such that the vector $L\mathbf{x}^{\Delta}$ has $(L\mathbf{x}^{\Delta})_{j^*} = -1$ and $(L\mathbf{x}^{\Delta})_{i^*} = \tau$ while all other entries are zero.
\end{claim}

\begin{proof}
For the first direction, assume $\tau \geq 0$ is a feasible tradeoff. Then, by the definition of a feasible tradeoff, there exists a vector $\mathbf{d}^{\Delta} \in \bbR^{\binom{T}{2}}$ with $d^{\Delta}_{j^*} < 0$ and $d^{\Delta}_{i^*} = - \tau d^{\Delta}_{j^*}$ and all other entries being zero, such that $\mathbf{d} + \mathbf{d}^{\Delta} \in \cD(G)$. Thus, there exists a solution $\mathbf{x}'$ of $A \mathbf{\tilde{x}} \geq \mathbf{b}$ such that $L(\mathbf{x}') = \mathbf{d} + \mathbf{d}^{\Delta}$. Let $\mathbf{x}^{\Delta} \defeq \mathbf{x}' - \mathbf{x}$. Then, by the linearity of $L$, we get that $L(\mathbf{x}^{\Delta}) = L(\mathbf{x}') - L(\mathbf{x}) = \mathbf{d}^{\Delta}$. Furthermore, since $\mathbf{x}'$ is a solution of $A \mathbf{\tilde{x}} \geq \mathbf{b}$, it is also a solution of $A_{EQ(\mathbf{x})} \mathbf{\tilde{x}} \geq \mathbf{b}_{EQ(\mathbf{x})}$, and thus
\[
 A_{EQ(\mathbf{x})} \mathbf{x}^{\Delta} = A_{EQ(\mathbf{x})} (\mathbf{x}' - \mathbf{x}) \geq \mathbf{b}_{EQ(\mathbf{x})} - A_{EQ(\mathbf{x})} \mathbf{x} = \mathbf{0}
\]
Since $(L\mathbf{x}^{\Delta})_{j^*} = -|d^{\Delta}_{j^*}|$ and $(L\mathbf{x}^{\Delta})_{i^*} = \tau |d^{\Delta}_{j^*}|$, the claim follows by scaling $\mathbf{x}^{\Delta}$.

For the other direction, assume there exists $\mathbf{x}^{\Delta}$ as in \cref{:cl:intermediate-claim}. Then, for every $\epsilon > 0$, it holds that $(L(\epsilon \mathbf{x}^{\Delta}))_{j^*} < 0$ and $(L(\epsilon \mathbf{x}^{\Delta}))_{i^*} = - \tau (L(\epsilon \mathbf{x}^{\Delta}))_{j^*}$ and that all other entries of $L(\epsilon \mathbf{x}^{\Delta})$ are zero. Furthermore, for each $\epsilon > 0$ it also holds that $A_{EQ(\mathbf{x})} (\epsilon \mathbf{x}^{\Delta}) \geq \mathbf{0}$, and thus
\[
 A_{EQ(\mathbf{x})} (\mathbf{x} + \epsilon \mathbf{x}^{\Delta}) \geq \mathbf{b}_{EQ(\mathbf{x})}
\]
We wish to choose $\epsilon$ such that $\mathbf{x} + \epsilon \mathbf{x}^{\Delta}$ is also a solution of the original program $A \mathbf{\tilde{x}} \geq \mathbf{b}$. To see that this is possible, notice that the constraints that appear in $A \mathbf{\tilde{x}} \geq \mathbf{b}$ but not in $A_{EQ(\mathbf{x})} \mathbf{\tilde{x}} \geq \mathbf{b}_{EQ(\mathbf{x})}$ are not tight at $\mathbf{x}$. Thus, by choosing a small enough $\epsilon > 0$, we get that $\mathbf{x} + \epsilon \mathbf{x}^{\Delta}$ is feasible for the original linear program, which means that $L(\mathbf{x} + \epsilon \mathbf{x}^{\Delta}) \in \cD(G)$. By the linearity of $L$ we get that $\mathbf{d} + L(\epsilon \mathbf{x}^{\Delta}) \in \cD(G)$, and thus, $\tau$ is a feasible tradeoff between $i^*$ and $j^*$ at the point $\mathbf{d}$.
\end{proof}

Using \cref{:cl:intermediate-claim}, we now construct a linear program for finding the maximum tradeoff $\tau_{\mathbf{d},i^*,j^*}$ at the point $\mathbf{d}$. For each $i \in \binom{T}{2}$, let $\tilde{x}_{i}$ denote the variable that $L$ projects onto commodity $i$. We start with the linear program $A_{EQ(\mathbf{x})} \mathbf{\tilde{x}} \geq \mathbf{0}$ and add the constraint $\tilde{x}_{j^*} = -1$ as well as a constraint $\tilde{x}_{i} = 0$ for every $i \in \binom{T}{2} \setminus \{i^*,j^*\}$. Lastly, we add the optimization objective to maximize $\tilde{x}_{i^*}$. Let $A' \mathbf{\tilde{x}} \geq \mathbf{b}'$ denote the new linear program. Notice that the width of the matrix $A'$ is $m_2$, and that all entries of $A'$ and $\mathbf{b}'$ are from $\{-1,0,1\}$.

Remember that one of the conditions of \cref{:prop:tradeoff-upper-bound} is that the maximum tradeoff $\tau_{\mathbf{d},i^*,j^*}$ is finite. Thus, by \cref{:thm:folklore-theorem}, there exists a submatrix $A''$ of $A'$ and a vector $\mathbf{b}''$ whose entries are a subset of the entries of $\mathbf{b}'$, such that the maximum tradeoff $\tau_{\mathbf{d},i^*,j^*}$ is the $i^*$th entry of the unique solution to the system $A'' \mathbf{\tilde{x}} = \mathbf{b}''$ of linear equalities. So, by Cramer's rule, $\tau_{\mathbf{d},i^*,j^*}$ is equal to the ratio of the determinants of two matrices whose entries are all from $\{-1,0,1\}$ and whose size is at most $m_2$. Such matrices must have integer determinants whose absolute value is at most $m_2!$, and thus $\tau_{\mathbf{d},i^*,j^*} \leq m_2!$. By \cref{:lem:existance-of-LP}, $m_2 \leq n^4$, and this concludes the proof of \cref{:prop:tradeoff-upper-bound}.

\section{Graph Family with Unbounded Tradeoff (Proof of \cref{:prop:unbounded-tradeoff-family})}
\label{:sec:unbounded-tradeoff-family}

This section proves \cref{:prop:unbounded-tradeoff-family}.
For convenience, we first recall its statement. 
\UnboundedTradeoffFamilyProposition*

At a high level, we construct a network $G_M$ 
where both commodities $i^*$ and $j^*$ have to cross a cut $(V_\myleft,V_\myright)$.
The idea is that the flow of commodity $j^*$ would have to use a path that crosses this cut $M$ times, while the flow of commodity $i^*$ only has to cross it once.
To force commodity $j^*$ to use such a path, we introduce additional commodities that saturate all the edges that connect this path to the rest of the network.

We now present the network $G_M$, assuming for simplicity that $M$ is an odd integer. The vertices of $G_M$ are naturally divided into two sets, which we denote by $V_\myleft$ and $V_\myright$.
Each side includes three terminals, which are denoted by $t_\mytopleft, t_\mymidleft ,t_\mybotleft$ and $t_\mytopright, t_\mymidright, t_\mybotright$.
The network includes a path of length $M$ from $t_\mymidleft$ to $t_\mymidright$, whose internal vertices are non-terminals, and whose vertices (including the endpoints) alternate between the two sides $V_\myleft$ and $V_\myright$.
Let $e_1,...,e_M$ denote the edges of this path. Thus, each of $e_1,...,e_M$ has one endpoint in $V_\myleft$ and one endpoint in $V_\myright$, which will be denoted by $\myleft(e_i)$ and $\myright(e_i)$.
For example, $\myleft(e_1) = t_\mymidleft$ and $\myright(e_1) = \myright(e_2)$.
The edges $e_1,...,e_M$ have capacity $2$. Furthermore, for each edge $e_i$ in this path, the network includes edges of capacity $1$ connecting $\myleft(e_i)$ to $t_\mytopleft$ and $t_\mybotleft$, and similarly connecting $\myright(e_i)$ to $t_\mytopright,t_\mybotright$.
Denote these edges by $e_{i,\mytopleft},e_{i,\mybotleft},e_{i,\mytopright}$ and $e_{i,\mybotright}$, respectively.
For convenience, we allow edges to appear with multiplicity, e.g., $e_{1,\mytopright}$ and $e_{2,\mytopright}$ are two parallel edges of capacity $1$ because $\myright(e_1)=\myright(e_2)$, although it could effectively be modeled by a single edge of capacity $2$.
Lastly, the network $G_M$ includes an edge $(t_\mytopleft,t_\mytopright)$ of capacity $2M$.
This concludes the description of the network $G_M$, see \cref{:fig:G_5} for illustration.

\begin{figure}[ht]
\includegraphics[width=0.5\textwidth]{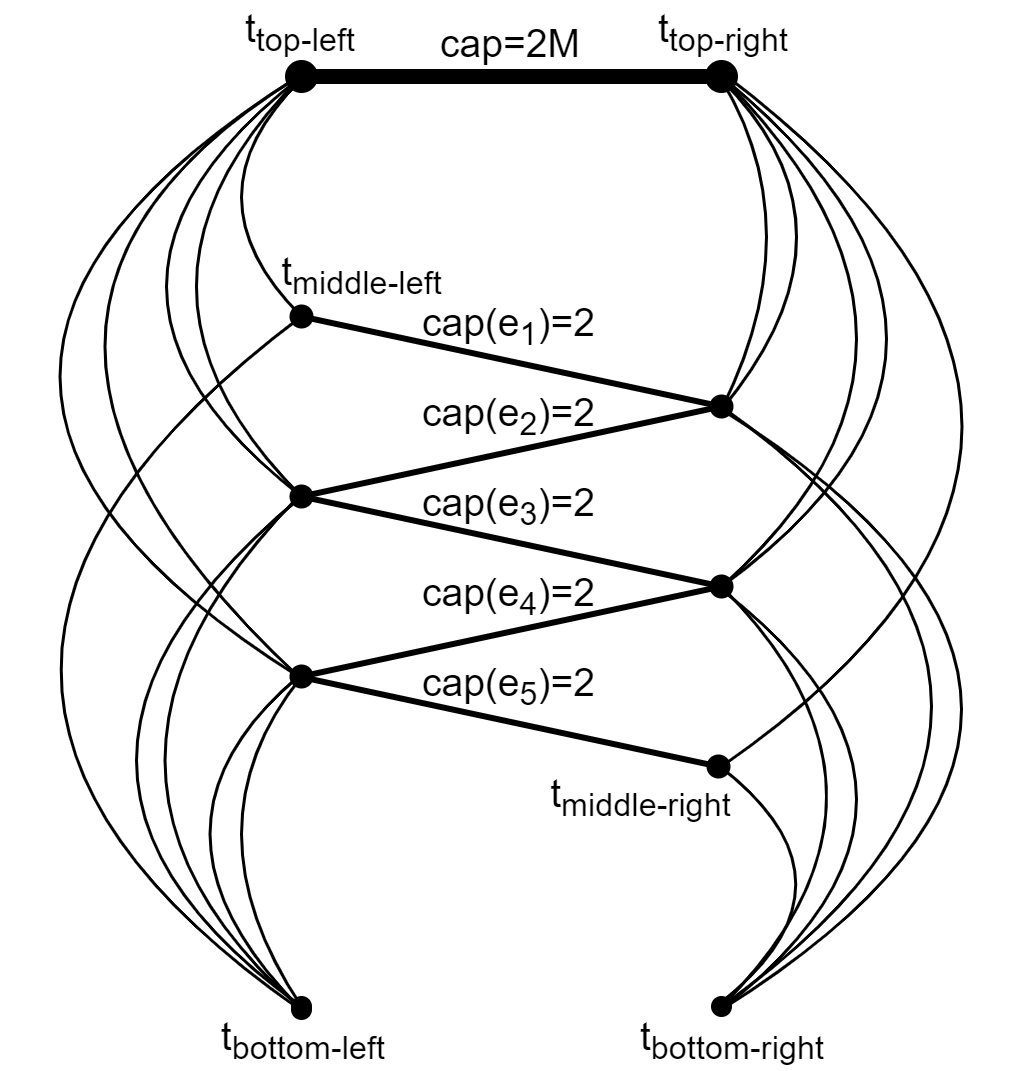}
\centering
\caption{The network $G_5$.}
\label{:fig:G_5}
\end{figure}

Let $i^*$ and $j^*$ be the commodities $i^* \defeq (t_\mytopleft,t_\mytopright)$ and $j^* \defeq (t_\mymidleft,t_\mymidright)$.
The demand vector $\mathbf{d}$ which we will be analyzing has a demand of $M$ for the commodity $(t_\mybotleft,t_\mytopright)$, a demand of $M$ for the commodity $(t_\mytopleft,t_\mybotright)$, a demand of $2$ for commodity $j^* = (t_\mymidleft,t_\mymidright)$, and a demand of zero for every other commodity. It is easy to see that this vector is in the demand polytope; as depicted in \cref{:subfig:first-flow}, we can realize $\mathbf{d}$ by sending $2$ units of flow along the path $(e_1,...,e_M)$, sending $M$ units of flow along the $M$ paths of the form $(t_\mybotleft, \myleft(e_i), t_\mytopleft,t_\mytopright)$, and sending $M$ units of flow along the $M$ paths of the form $(t_\mybotright, \myright(e_i), t_\mytopright,t_\mytopleft)$. These $2M$ paths all overlap on the edge $(t_\mytopleft,t_\mytopright)$, but this is fine because this edge has capacity $2M$.

In order to prove \cref{:prop:unbounded-tradeoff-family}, we need to show that $M$ is a feasible tradeoff between commodities $i^*$ and $j^*$ at the point $\mathbf{d}$, as well as show that the maximum tradeoff $\tau_{\mathbf{d},i^*,j^*}$ is finite.

Let $\mathbf{d}^{\Delta} \in \bbR^{\binom{T}{2}}$ be the vector such that $d^{\Delta}_{j^*} = -2$ and $d^{\Delta}_{i^*} = - M d^{\Delta}_{j^*} = 2M$, and every other entry of $\mathbf{d}^{\Delta}$ is zero. To show that $M$ is a feasible tradeoff between commodities $i^*$ and $j^*$ at the point $\mathbf{d}$, we need to show that $\mathbf{d}' \defeq \mathbf{d} + \mathbf{d}^{\Delta}$ is in the demand polytope $\cD(G_M)$. To get some intuition, let us start by explicitly stating the demand vector $\mathbf{d}'$:
It has a demand of $2M$ for commodity $i^* = (t_\mytopleft,t_\mytopright)$, a demand of $M$ for the commodity $(t_\mybotleft,t_\mytopright)$, a demand of $M$ for the commodity $(t_\mytopleft,t_\mybotright)$, and zero for all other commodities.
Now, it is not hard to see that $\mathbf{d}'$ is indeed in the demand polytope; as depicted in \cref{:subfig:second-flow}, it can be realized by sending the $2M$ units of flow of commodity $i^*$ along the edge $(t_\mytopleft,t_\mytopright)$ of capacity $2M$, while sending $M$ units of flow along the $M$ paths of the form $(t_\mybotleft,\myleft(e_i),\myright(e_i),t_\mytopright)$ and sending $M$ units of flow along the $M$ paths of the form $(t_\mybotright,\myright(e_i),\myleft(e_i),t_\mytopleft)$.

\begin{figure}[ht]
     \centering
     \begin{subfigure}[b]{0.49\textwidth}
         \centering
         \includegraphics[width=\textwidth]{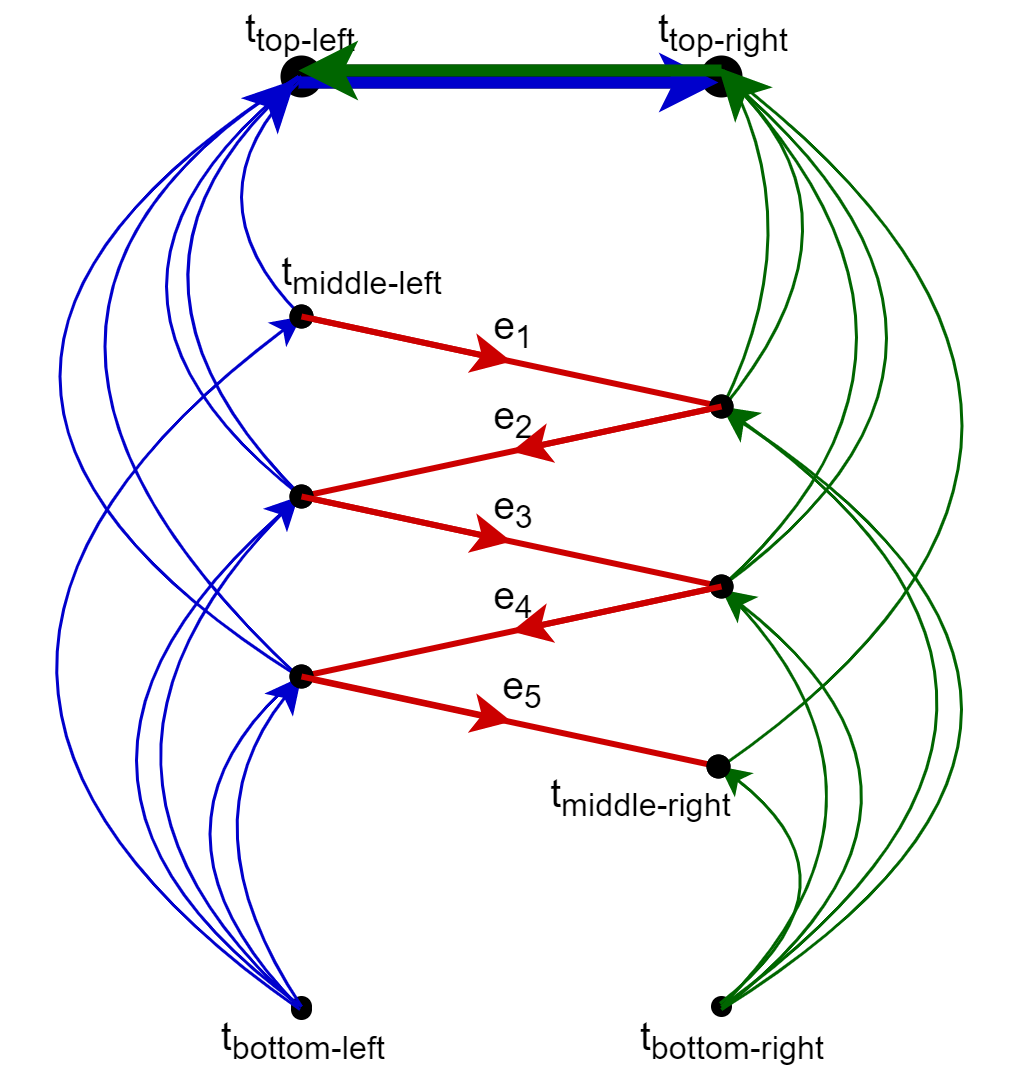}
         \caption{Demand $\mathbf{d}$.}
         \label{:subfig:first-flow}
     \end{subfigure}
     \hfill
     \begin{subfigure}[b]{0.49\textwidth}
         \centering
         \includegraphics[width=\textwidth]{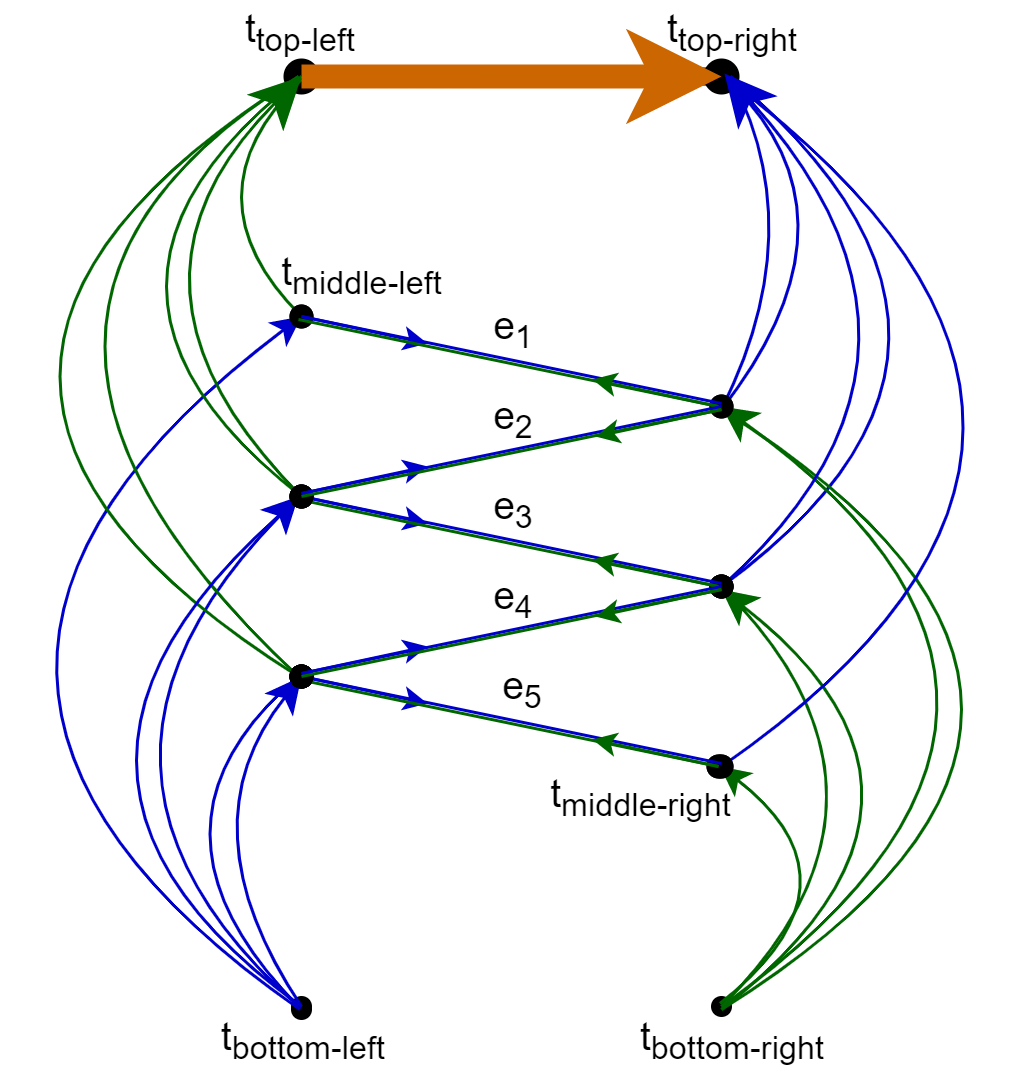}
         \caption{Demand $\mathbf{d}' = \mathbf{d} + \mathbf{d}^{\Delta} $}
         \label{:subfig:second-flow}
     \end{subfigure}
        \caption{Multicommodity flows that realize the demand vectors $\mathbf{d}$ and $\mathbf{d}'$, using colors to distinguish between the commodities. }
        \label{:fig:flows}
\end{figure}

In order to complete the proof of \cref{:prop:unbounded-tradeoff-family}, we just need to show that the maximum tradeoff $\tau_{\mathbf{d},i^*,j^*}$ is finite. To do this, we will use the following claim to show that every feasible tradeoff $\tau$ must satisfy $\tau \leq M$.

\begin{claim}
Let $\mathbf{d}''$ be a demand vector in the the demand polytope $\cD(G_M)$ that has a demand of $M$ for the commodity $(t_\mybotleft,t_\mytopright)$, and a demand of $M$ for the commodity $(t_\mytopleft,t_\mybotright)$; the other entries in $\mathbf{d}''$ can be arbitrary. Then, it must hold that $d''_{i^*} + M \cdot d''_{j^*} \leq 2M$, which means that
\[
     (-1) \cdot \frac{d''_{i^*} - d_{i^*}}{d''_{j^*} - d_{j^*}} = \frac{d''_{i^*}}{2 - d''_{j^*}} \leq M .
\]
\end{claim}

\begin{proof}
Fix some multicommodity flow that realizes the demand vector $\mathbf{d}''$.
Let $S_\mytop \defeq \{t_\mytopleft,t_\mytopright\}$.
Since the cut $(S_\mytop,\overline{S_\mytop})$ separates the endpoints of each of the two commodities $(t_\mybotleft,t_\mytopright)$ and $(t_\mytopleft,t_\mybotright)$, and since the total capacity of this cut is equal to the sum of the demands on these two commodities, we get that all edges crossing this cut must be saturated by the flows of the two commodities $(t_\mybotleft,t_\mytopright)$ and $(t_\mytopleft,t_\mybotright)$.
Thus, the flow of commodity $j^*$ cannot use any of the edges that cross this cut.
By letting $S_\mybot \defeq \{t_\mybotleft,t_\mybotright\}$ and making the exact same argument, we get that the flow of commodity $j^*$ cannot use any of the edges that cross the cut $(S_\mybot,\overline{S_\mybot})$.
However, the only path between the endpoints of commodity $j^* = (t_\mymidleft,t_\mymidright)$ that does not use any edges that cross either of the two aforementioned cuts is the path $(e_1,...,e_M)$. This means that every unit of flow of commodity $j^*$ must pass $M$ times through the cut $(V_\myleft,V_\myright)$, and therefore the flow of commodity $j^*$ must use up a total of at least $M \cdot d''_{j^*}$ capacity from edges crossing the cut $(V_\myleft,V_\myright)$. Furthermore, because the cut $(V_\myleft,V_\myright)$ separates the endpoints of each of the other three commodities, the flows of commodities $(t_\mybotleft,t_\mytopright)$ and $(t_\mytopleft,t_\mybotright)$ must each use at least $M$ capacity from edges crossing this cut, and the flow of commodity $i^*$ must use at least $d''_{i^*}$ capacity from edges crossing this cut. But the total capacity of the cut $(V_\myleft,V_\myright)$ is $4M$, and thus
\[
 M \cdot d''_{j^*} + M + M + d''_{i^*} \leq 4M ,
\]
which gives us the desired inequality $d''_{i^*} + M \cdot d''_{j^*} \leq 2M$.
\end{proof}
This completes the proof of \cref{:prop:unbounded-tradeoff-family}.

\section{Refutation of Seymour's conjectures (Proof of \cref{:thm:refutation-of-seymour's-conjectures})}
\label{:sec:proofSeymour}

We now prove one direction of \cref{:cl:alternative-characterization-of-tradeoff} 
(namely, that statement 2 implies statement 1), and then use it to prove \cref{:thm:refutation-of-seymour's-conjectures}. 
We omit the other direction of the claim 
because our proof of \cref{:thm:refutation-of-seymour's-conjectures} 
does not use it, at least not in a blackbox manner;
instead, our proof follows its steps implicitly
and picks a vector $\mathbf{d}^*$ with additional properties
not promised by \cref{:cl:alternative-characterization-of-tradeoff}
(namely, having integral entries and a small number of non-zero entries). 

\AlternativeCharacterizationOfTradeoff*

\begin{proof} [Proof of \cref{:cl:alternative-characterization-of-tradeoff} (the direction $(2) \rightarrow (1)$)]
Since the points $\mathbf{d}^* \notin \cD(G)$ and $\mathbf{\hat{d}} \in \cD(G)$ differ only in the $i^*$th coordinate, the line between them must intersect some facet $\sum_{i \in \binom{V}{2}} a_i \tilde{d}_i \leq b$ of the polytope $\cD(G)$ for which $a_{i^*} > 0$. Let $\mathbf{d}$ be the point at which this intersection occurs. Then, $\mathbf{d}$ is in the polytope, and the maximum tradeoff $\tau_{\mathbf{d},i^*,j^*}$ is at most $\frac{a_{j^*}}{a_{i^*}}$, which in particular means it is finite. Furthermore, since $\mathbf{d}$ is on the line between $\mathbf{d}^*$ and $\mathbf{\hat{d}}$, it differs from these two points only in the $i^*$th coordinate, and by the down monotonicity of $\cD(G)$ this implies that $d_{i^*} < d^*_{i^*}$. Thus, $\mathbf{d}$ must have a feasible tradeoff larger than $\tau$; by defining $\mathbf{d}^{\Delta} \defeq \mathbf{d}^* + \mathbf{d}^{\Delta*} - \mathbf{d}$, we get that $\mathbf{d} + \mathbf{d}^{\Delta} = \mathbf{d}^* + \mathbf{d}^{\Delta*} \in \cD(G)$ and that $d^{\Delta}_{i^*} > d^{\Delta*}_{i^*}$ while every other coordinate of $\mathbf{d}^{\Delta}$ is the same as in $\mathbf{d}^{\Delta*}$. This means that $\tau_{\mathbf{d},i^*,j^*} > \tau$, which concludes the proof.
\end{proof}

\RefutationOfSeymoursConjectures*

To keep things simple, we first prove the theorem while dropping the condition $\norm{\mathbf{d}^*}_0 \leq 3$, which is equivalent to refuting the first conjecture presented in \cref{:subsec:seymour's-conjectures}, and then show how to adapt the proof to satisfy the condition $\norm{\mathbf{d}^*}_0 \leq 3$, thereby refuting the second conjecture in \cref{:subsec:seymour's-conjectures}.

\begin{proof} [Proof of \cref{:thm:refutation-of-seymour's-conjectures} (without the condition $\norm{\mathbf{d}^*}_0 \leq 3$)]
Fix any integer $N > 0$.
Our goal is to provide an instance of \cref{:prob:seymour's-problem} in which the demand polytope of every compatible network must include a point with a finite maximum tradeoff greater than $(N^4)!$.
Let $M \defeq (N^4)! + 1$, and let $G_{M}$ be the network from our proof of \cref{:prop:unbounded-tradeoff-family}.
Furthermore let $i^*$, $j^*$, $\mathbf{d}$, $\mathbf{d}^{\Delta}$ be as in that proof.
Let the vector $\mathbf{d}^*$ be the result of increasing coordinate $i^*$ of $\mathbf{d}$ by $1$, and let ${\mathbf{d}^{\Delta*}} \defeq \mathbf{d} + \mathbf{d}^{\Delta} - \mathbf{d}^*$.
Thus, ${d^{\Delta*}}_{j^*} = -2$ and ${d^{\Delta*}}_{i^*} = 2M - 1 > 2 \cdot (N^4)!$, and every other entry of ${\mathbf{d}^{\Delta*}}$ is zero. Furthermore, $\mathbf{d}^* \notin \cD(G_{M})$.
Notice that according to the construction of $G_{M}$ and $\mathbf{d}^*$, the capacities of $G_{M}$ and the entries of $\mathbf{d}^*$ are all integral, and $\norm{\mathbf{d}^*}_0 = 4$.
Now consider the instance of \cref{:prob:seymour's-problem} comprising of these network $G_{M}$ and demand $\mathbf{d}^*$.
Every compatible network $G'$ of this instance must satisfy $\mathbf{d}^* \notin \cD(G')$.
Furthermore, as $G'$ results from $G_{M}$ by edge-contraction operations, it must satisfy $\cD(G_{M}) \subseteq \cD(G')$ and thus $\mathbf{d} \in \cD(G')$ and $\mathbf{d}^* + \mathbf{d}^{\Delta*} = \mathbf{d} + \mathbf{d}^{\Delta} \in \cD(G')$.
Using \cref{:cl:alternative-characterization-of-tradeoff} with $\mathbf{\hat{d}} = \mathbf{d}$, every compatible network $G'$ must include a point that has a finite maximum tradeoff larger than $\frac{{d^{\Delta*}}_{i^*}}{|{d^{\Delta*}}_{j^*}|} > \frac{2 \cdot (N^4)!}{2}$, and the theorem follows using \cref{:prop:tradeoff-upper-bound}.
\end{proof}

\begin{proof} [Proof of \cref{:thm:refutation-of-seymour's-conjectures} (with the condition $\norm{\mathbf{d}^*}_0 \leq 3$)]
Fix $N$, and let $(G_M, \mathbf{d}^*)$ be the instance of \cref{:prob:seymour's-problem} from the previous proof. Thus, every compatible network $G'$ of $(G_M, \mathbf{d}^*)$ must have at least $N$ vertices. However, the condition $\norm{\mathbf{d}^*}_0 \leq 3$ does not hold, in fact $\norm{\mathbf{d}^*}_0 = 4$.

Remove one unit from the capacity of the edge $(t_\mytopleft,t_\mytopright)$ in $G_M$ as well as one unit from the demand of commodity $i^*=(t_\mytopleft,t_\mytopright)$ in $\mathbf{d}^*$, and let $(G^{M}_{-}, \mathbf{d}^*_{-})$ be the resulting instance. The fact that $\mathbf{d}^*$ is not feasible in $G_M$ implies that $\mathbf{d}^*_{-}$ is not feasible in $G^{M}_{-}$, so this is indeed a valid instance of \cref{:prob:seymour's-problem}.
Furthermore, $\norm{\mathbf{d}^*_{-}}_0 \leq 3$ holds.
It thus remains to show that every compatible network $G'_{-}$ of the instance $(G^{M}_{-}, \mathbf{d}^*_{-})$ must have at least $N$ vertices.

Let $G'_{-}$ be any compatible network of the instance $(G^{M}_{-}, \mathbf{d}^*_{-})$. This implies that $G'_{-}$ is formed from $G^{M}_{-}$ using contractions, and that $\mathbf{d}^*_{-}$ is infeasible in $G'_{-}$. Let $t_\mytopleft'$ and $t_\mytopright'$ be the vertices in $G'_{-}$ that correspond to the vertices $t_\mytopleft$ and $t_\mytopright$ in $G^{M}_{-}$. The vertices $t_\mytopleft'$ and $t_\mytopright'$ may have been formed from multiple vertices of $G^{M}_{-}$ through contractions, and may even be the same vertex (i.e. if $t_\mytopleft$ and $t_\mytopright$ were merged together by contractions).
By adding $1$ to the capacity of the edge $(t_\mytopleft',t_\mytopright')$ in $G'_{-}$, we get a network $G'$ that can be created from $G_M$ using contractions, and satisfies that $d^{*}$ is infeasible in $G'$. Thus, $G'$ is a compatible network of $(G_M, \mathbf{d}^*)$, which means it has at least $N$ vertices. The theorem follows since $G_{-}'$ has the same number of vertices as $G'$.
\end{proof}

\section{Future Directions} \label{:sec:further-directions}

Our main result shows that there is no function $f:\bbN \to \bbN$ such that every $k$-terminal network $G$ admits an exact flow sparsifier of size bounded by $f(k)$; in fact, this is not possible even for $k=6$.
We present here three interesting directions for future work.

\begin{direction}
    Extend our main result to planar networks $G$ (and other graph families).
\end{direction}
Our framework offers a potential method to prove such a result; all one has to do is show that planar networks can have arbitrarily large finite tradeoffs (i.e., extend \cref{:prop:unbounded-tradeoff-family}).
A weaker version of such a result could restrict the sparsifier graph to be planar as well.
We point out that \cref{:thm:main-result} holds for input graphs excluding $K_7$ as a minor, and one can try to determine the smallest excluded minor for such results.

\begin{direction}
Extend from exact sparsifiers to quality $1+\epsilon$.
That is, provide bounds on $f(k,\epsilon)$,
the smallest value for which every $k$-terminal network $G$ admits
a $(1+\epsilon)$-quality flow sparsifier $G'$ of size at most $f(k,\epsilon)$.
The quality bound here means that $G'$ has the same terminals as $G$
and satisfies $\cD(G) \subseteq \cD(G') \subseteq (1+\epsilon) \cD(G)$.
\end{direction}
Our results imply a lower bound on $f(k,\epsilon)$, although it is somewhat weak, namely,
$f(k,\epsilon) = \Omega\Big( \Big( \frac{\log (1/\epsilon)}{\log \log (1/\epsilon)} \Big)^{1/4} \Big)$ for every $k \geq 6$ 
and $\epsilon \in (0,\frac{1}{4})$.
This can be strengthened by an $\Omega(k)$ factor,
by simply taking $k/6$ copies of a network with $6$ terminals.
Using different techniques, we can prove stronger lower bounds $f(k,\epsilon) = \tilde{\Omega}(\frac{k}{\epsilon^{1/6}})$ and $f(k,\epsilon) = \tilde{\Omega}(\min\{2^{k/8},\frac{k}{\epsilon^{1/2}}\})$ for every $k \geq 10$ and $\epsilon \in (0,\frac{1}{2})$, where the notation $\tilde{\Omega}(\cdot)$ hides $\textrm{polylog} (1/\epsilon)$ factors.
In fact, it seems quite challenging to get any upper bound on $f(k,\epsilon)$, 
and this gap remains a fascinating open question. 

\begin{direction}
    Extend the result to $k=5$.
\end{direction}

It is known that for every $k \leq 4$,
there exists $N_k' > 0$ such that every $k$-terminal network $G$ has an exact flow sparsifier of size at most $N_k'$, see e.g.~\cite{AGK14}.
Our main result shows that this is not possible for any $k \geq 6$, but it remains open for $k=5$.

{\small 
\bibliographystyle{alphaurl}
\bibliography{robi}
}

\appendix

\section{Proof of \cref{:lem:existance-of-LP}}
\label{:subsec:linear-program-for-flow}

In this section we present a linear program with the properties promised by \cref{:lem:existance-of-LP}.
The linear program that we use is the most standard linear program for the multicommodity flow problem. For each commodity $i \in \binom{T}{2}$ we have a variable $\tilde{d}_i$ representing the number of units of flow shipped by the flow of the $i$-th commodity, as well as variables $\tilde{f}_{i,e}$ representing the value of $f_i$ at each edge $e \in E$. Then, we have a constraint
\[
 \tilde{d}_i \geq 0 ,
\]
and for each vertex $v \in V$, we have a flow-conservation constraint for the flow $f_i$ at the vertex $v$. If $v \notin \{source(i),sink(i)\}$ then the flow-conservation constraint is
\[
 \sum_{\text{$e$ entering $v$}} \tilde{f}_{i,e} + \sum_{\text{$e$ leaving $v$}} (-1) \cdot \tilde{f}_{i,e} = 0 .
\]
However, if $v = source(i)$ then the flow-conservation constraint is
\[
 \tilde{d}_i + \sum_{\text{$e$ entering $v$}} \tilde{f}_{i,e} + \sum_{\text{$e$ leaving $v$}} (-1) \cdot \tilde{f}_{i,e} = 0 ,
\]
and if $v = sink(i)$ then the constraint is
\[
 (-1) \cdot \tilde{d}_i + \sum_{\text{$e$ entering $v$}} \tilde{f}_{i,e} + \sum_{\text{$e$ leaving $v$}} (-1) \cdot \tilde{f}_{i,e} = 0
\]
Lastly, aside from including each of the aforementioned variables and constraints for each commodity $i \in \binom{T}{2}$, we need to include constraints that force the multicommodity flow to respect the capacities of the edges. To do this, for each edge $e \in E$, and each possible choice of signs for the commodities $\{s_i \in \{1,-1\}\}_{i \in \binom{T}{2}}$, we add a constraint
\[
 \sum_{i \in \binom{T}{2}} s_i \cdot \tilde{f}_{i,e} \leq c(e) .
\]
Observe that these constraints are together equivalent to the non-linear constraint 
$\sum_{i \in \binom{T}{2}} |\tilde{f}_{i,e}| \leq c(e)$ that we need to enforce.

This concludes the presentation of the linear program. It is easy to see that the solutions of the program are in one-to-one correspondence to the multicommodity flows in the network, and thus that projecting the feasible region of the program onto the variables $\{\tilde{d}_i\}_{i \in \binom{T}{2}}$ gives the demand polytope $\cD(G)$. Furthermore, the coefficients of the variables in the program are all from $\{-1,0,1\}$. Lastly, since we have exactly $(1 + |E|)$ variables for each commodity $i \in \binom{T}{2}$, we get that the total number of variables in the program is
\[
 (1 + |E|) \cdot \binom{T}{2} \leq n^2 \cdot |T|^2 \leq n^4 .
\]

This concludes the proof of  \cref{:lem:existance-of-LP}.
\end{document}

%% file: defs.tex
\usepackage{amsmath}
\usepackage{amsfonts}
\usepackage{mathtools}
\usepackage{amssymb}
\usepackage{amsthm}
\usepackage{enumitem}
\usepackage{thmtools, thm-restate}
\usepackage{graphicx}
\usepackage{color}
\usepackage{hyperref}
\hypersetup{colorlinks=true,allcolors=blue}
\usepackage[capitalise,noabbrev]{cleveref}
\usepackage{caption}
\usepackage{subcaption}
\usepackage{comment}

\newtheorem{theorem}{Theorem}[section] 

\newtheorem{lemma}[theorem]{Lemma}

\newtheorem{claim}[theorem]{Claim}

\newtheorem{direction}{Direction}

\theoremstyle{definition}
\newtheorem{definition}[theorem]{Definition}

\def\cD{{\cal D}}

\def\bbN{{\mathbb N}}

\def\bbR{{\mathbb R}}

\newcommand{\norm}[1]{\left\| {#1} \right\|}

\newcommand*{\defeq}{\stackrel{\text{def}}{=}}
\newcommand{\quotes}[1]{{``{#1}"}}

\newcommand{\myleft}{\textrm{left}}
\newcommand{\myright}{\textrm{right}}
\newcommand{\mytop}{\textrm{top}}
\newcommand{\mybot}{\textrm{bottom}}
\newcommand{\mytopleft}{\textrm{top\_left}}
\newcommand{\mytopright}{\textrm{top\_right}}
\newcommand{\mymidleft}{\textrm{middle\_left}}
\newcommand{\mymidright}{\textrm{middle\_right}}
\newcommand{\mybotleft}{\textrm{bottom\_left}}
\newcommand{\mybotright}{\textrm{bottom\_right}}

%% file: main.bbl
\newcommand{\etalchar}[1]{$^{#1}$}
\begin{thebibliography}{CDK{\etalchar{+}}21}

\bibitem[AGK14]{AGK14}
A.~Andoni, A.~Gupta, and R.~Krauthgamer.
\newblock Towards $(1+\epsilon)$-approximate flow sparsifiers.
\newblock In {\em 25th Annual ACM-SIAM Symposium on Discrete Algorithms}, pages
  279--293, 2014.
\newblock \href {https://doi.org/10.1137/1.9781611973402.20}
  {\path{doi:10.1137/1.9781611973402.20}}.

\bibitem[BK22]{BK22}
Itai Boneh and Robert Krauthgamer.
\newblock Optimal vertex-cut sparsification of quasi-bipartite graphs.
\newblock {\em CoRR}, abs/2207.01459, 2022.
\newblock \href {http://arxiv.org/abs/2207.01459} {\path{arXiv:2207.01459}}.

\bibitem[CDK{\etalchar{+}}21]{CDLKLPSV21}
Parinya Chalermsook, Syamantak Das, Yunbum Kook, Bundit Laekhanukit, Yang~P.
  Liu, Richard Peng, Mark Sellke, and Daniel Vaz.
\newblock Vertex sparsification for edge connectivity.
\newblock In {\em 32nd Annual ACM-SIAM Symposium on Discrete Algorithms},
  SODA'21, page 1206–1225. SIAM, 2021.
\newblock \href {https://doi.org/10.1137/1.9781611976465.74}
  {\path{doi:10.1137/1.9781611976465.74}}.

\bibitem[Chu12]{Chuzhoy12}
Julia Chuzhoy.
\newblock On vertex sparsifiers with {S}teiner nodes.
\newblock In {\em 44th symposium on Theory of Computing}, pages 673--688. ACM,
  2012.
\newblock \href {https://doi.org/10.1145/2213977.2214039}
  {\path{doi:10.1145/2213977.2214039}}.

\bibitem[Chu16]{Chuzhoy16}
Julia Chuzhoy.
\newblock Routing in undirected graphs with constant congestion.
\newblock {\em SIAM Journal on Computing}, 45(4):1490--1532, 2016.

\bibitem[CLLM10]{CLLM10}
Moses Charikar, Tom Leighton, Shi Li, and Ankur Moitra.
\newblock Vertex sparsifiers and abstract rounding algorithms.
\newblock In {\em 51st Annual Symposium on Foundations of Computer Science},
  pages 265--274. IEEE Computer Society, 2010.
\newblock \href {https://doi.org/10.1109/FOCS.2010.32}
  {\path{doi:10.1109/FOCS.2010.32}}.

\bibitem[CSWZ00]{CSWZ00}
S.~Chaudhuri, K.~V. Subrahmanyam, F.~Wagner, and C.~D. Zaroliagis.
\newblock Computing mimicking networks.
\newblock {\em Algorithmica}, 26:31--49, 2000.
\newblock \href {https://doi.org/10.1007/s004539910003}
  {\path{doi:10.1007/s004539910003}}.

\bibitem[EGK{\etalchar{+}}14]{EGKRTT14}
M.~Englert, A.~Gupta, R.~Krauthgamer, H.~R{\"a}cke, I.~Talgam-Cohen, and
  K.~Talwar.
\newblock Vertex sparsifiers: New results from old techniques.
\newblock {\em SIAM Journal on Computing}, 43(4):1239--1262, 2014.
\newblock \href {http://arxiv.org/abs/1006.4586} {\path{arXiv:1006.4586}},
  \href {https://doi.org/10.1137/130908440} {\path{doi:10.1137/130908440}}.

\bibitem[GHP20]{GHP20}
Gramoz Goranci, Monika Henzinger, and Pan Peng.
\newblock Improved guarantees for vertex sparsification in planar graphs.
\newblock {\em SIAM Journal on Discrete Mathematics}, 34(1):130--162, 2020.
\newblock \href {https://doi.org/10.1137/17M1163153}
  {\path{doi:10.1137/17M1163153}}.

\bibitem[GR16]{GR16}
Gramoz Goranci and Harald R{\"a}cke.
\newblock Vertex sparsification in trees.
\newblock In {\em International Workshop on Approximation and Online
  Algorithms}, pages 103--115. Springer, 2016.

\bibitem[HKNR98]{HKNR98}
Torben Hagerup, Jyrki Katajainen, Naomi Nishimura, and Prabhakar Ragde.
\newblock Characterizing multiterminal flow networks and computing flows in
  networks of small treewidth.
\newblock {\em J. Comput. Syst. Sci.}, 57:366--375, 1998.
\newblock \href {https://doi.org/10.1006/jcss.1998.1592}
  {\path{doi:10.1006/jcss.1998.1592}}.

\bibitem[HLW21]{HLW21}
Zhiyang He, Jason Li, and Magnus Wahlstr{\"{o}}m.
\newblock Near-linear-time, optimal vertex cut sparsifiers in directed acyclic
  graphs.
\newblock In {\em 29th Annual European Symposium on Algorithms, {ESA} 2021},
  volume 204 of {\em LIPIcs}, pages 52:1--52:14. Schloss Dagstuhl -
  Leibniz-Zentrum f{\"{u}}r Informatik, 2021.
\newblock \href {https://doi.org/10.4230/LIPIcs.ESA.2021.52}
  {\path{doi:10.4230/LIPIcs.ESA.2021.52}}.

\bibitem[KPZ19]{KPZ19}
Nikolai Karpov, Marcin Pilipczuk, and Anna Zych{-}Pawlewicz.
\newblock An exponential lower bound for cut sparsifiers in planar graphs.
\newblock {\em Algorithmica}, 81(10):4029--4042, 2019.
\newblock \href {https://doi.org/10.1007/s00453-018-0504-8}
  {\path{doi:10.1007/s00453-018-0504-8}}.

\bibitem[KR13]{KR13}
Robert Krauthgamer and Inbal Rika.
\newblock Mimicking networks and succinct representations of terminal cuts.
\newblock In {\em 24th Annual ACM-SIAM Symposium on Discrete Algorithms}, pages
  1789--1799. SIAM, 2013.
\newblock \href {https://doi.org/10.1137/1.9781611973105.128}
  {\path{doi:10.1137/1.9781611973105.128}}.

\bibitem[KR14]{KR14}
Arindam Khan and Prasad Raghavendra.
\newblock On mimicking networks representing minimum terminal cuts.
\newblock {\em Information Processing Letters}, 114(7):365--371, 2014.
\newblock \href {https://doi.org/10.1016/j.ipl.2014.02.011}
  {\path{doi:10.1016/j.ipl.2014.02.011}}.

\bibitem[KR20]{KR20}
Robert Krauthgamer and Havana~Inbal Rika.
\newblock Refined vertex sparsifiers of planar graphs.
\newblock {\em SIAM Journal on Discrete Mathematics}, 34(1):101--129, 2020.
\newblock \href {https://doi.org/10.1137/17M1151225}
  {\path{doi:10.1137/17M1151225}}.

\bibitem[KW20]{KW20}
Stefan Kratsch and Magnus Wahlstr\"{o}m.
\newblock Representative sets and irrelevant vertices: New tools for
  kernelization.
\newblock {\em J. ACM}, 67(3), 2020.
\newblock \href {https://doi.org/10.1145/3390887} {\path{doi:10.1145/3390887}}.

\bibitem[Liu20]{Liu20}
Yang~P. Liu.
\newblock Vertex sparsification for edge connectivity in polynomial time.
\newblock {\em CoRR}, abs/2011.15101, 2020.
\newblock \href {http://arxiv.org/abs/2011.15101} {\path{arXiv:2011.15101}}.

\bibitem[LM10]{LM10}
F.~Thomson Leighton and Ankur Moitra.
\newblock Extensions and limits to vertex sparsification.
\newblock In {\em 42nd ACM symposium on Theory of computing}, STOC, pages
  47--56. ACM, 2010.
\newblock \href {https://doi.org/10.1145/1806689.1806698}
  {\path{doi:10.1145/1806689.1806698}}.

\bibitem[MM16]{MM16}
Konstantin Makarychev and Yury Makarychev.
\newblock Metric extension operators, vertex sparsifiers and {L}ipschitz
  extendability.
\newblock {\em Israel Journal of Mathematics}, 212(2):913--959, 2016.
\newblock \href {https://doi.org/10.1007/s11856-016-1315-8}
  {\path{doi:10.1007/s11856-016-1315-8}}.

\bibitem[Moi09]{Moitra09}
Ankur Moitra.
\newblock Approximation algorithms for multicommodity-type problems with
  guarantees independent of the graph size.
\newblock In {\em 50th Annual Symposium on Foundations of Computer Science},
  FOCS, pages 3--12. IEEE, 2009.
\newblock \href {https://doi.org/10.1109/FOCS.2009.28}
  {\path{doi:10.1109/FOCS.2009.28}}.

\bibitem[Moi11]{Moitra11}
Ankur Moitra.
\newblock {\em Vertex sparsification and universal rounding algorithms}.
\newblock PhD thesis, Massachusetts Institute of Technology, 2011.
\newblock URL: \url{https://dspace.mit.edu/handle/1721.1/66019}.

\bibitem[RBC07]{RBC07}
G{\'a}bor R{\'e}tv{\'a}ri, J{\'o}zsef~J. B{\'\i}r{\'o}, and Tibor Cinkler.
\newblock Fairness in capacitated networks: A polyhedral approach.
\newblock In {\em IEEE INFOCOM 2007-26th IEEE International Conference on
  Computer Communications}, pages 1604--1612. IEEE, 2007.
\newblock \href {https://doi.org/10.1109/INFCOM.2007.188}
  {\path{doi:10.1109/INFCOM.2007.188}}.

\bibitem[Sey15]{Seymour15}
Paul Seymour.
\newblock Criticality for multicommodity flows.
\newblock {\em Journal of Combinatorial Theory, Series B}, 110:136--179, 2015.
\newblock \href {https://doi.org/10.1016/j.jctb.2014.08.001}
  {\path{doi:10.1016/j.jctb.2014.08.001}}.

\end{thebibliography}
